\documentclass[12pt]{article}
\usepackage{amsmath, amsthm, amsfonts, amssymb}
\usepackage[margin=1in]{geometry}
\usepackage{setspace}
\usepackage{natbib}

\newtheorem{theorem}{Theorem}
\newtheorem{definition}{Definition}
\newtheorem{assumption}{Assumption}

\newtheorem{corollary}{Corollary}
\newtheorem{proposition}{Proposition}
\newtheorem{remark}{Remark}

\newcommand{\E}{\mathbb{E}}

\newcommand{\F}{\mathcal{F}}

\newcommand{\Z}{\mathcal{Z}}

\author{Eduardo Fonseca Mendes\thanks{The author acknowledges financial support from the Fundação de Amparo à Pesquisa do Estado de São Paulo (FAPESP), grant number 2020/07019-2.}\\{\small Getulio Vargas Foundation}}
\date{}
\title{Consistent Order Selection under Non-Identifiability and Dependence}
\begin{document}
\maketitle
\begin{abstract}
We provide sufficient conditions for the consistency of penalized least squares procedures that select the order (dimension) of a regression model from a sequence of nested classes, allowing for dependent, martingale-difference errors. The main contribution is to relax the classical identifiability requirement: parameters indexing classes larger than the true order need not be identified, provided the additional, excess directions admit a linear approximation to the truth in a neighbourhood of the true parameter. This relaxation lets the number of candidate models grow with the sample size, removing the usual need for a fixed upper bound. We verify the resulting high-level conditions for two classes of nonlinear regression models used in applied work: multiple-regime smooth transition regression and mixture-of-experts models with generic generalized linear model experts. Under a BIC-type penalty, the resulting order selection rule is consistent in either model class whenever the number of candidate models grows slower than the logarithm of the sample size.
\end{abstract}

\noindent Key words and phrases: Model order selection; Non-identifiability; Penalized least squares; Martingale differences; Mixture-of-experts models; Smooth transition regression models.

\section{Introduction}
Choosing the number of regimes or components in a nonlinear regression model is a recurring problem in nonlinear time series analysis and beyond. Smooth transition and threshold autoregressions, finite mixtures, and mixtures of experts all nest a well-identified low-order model inside a strictly larger family that is no longer identified once the true order $m_0$ is exceeded. Extra regimes can be relabelled, collapsed onto an existing regime, or have their transition-function slope driven to zero without changing the regression function, so the parameter indexing them is not uniquely determined by the data. In this case, the classical asymptotic theory behind AIC \citep{akaike1974}, BIC \citep{schwarz1978}, and related penalized criteria fails to apply for $m>m_0$.

Statistical literature has addressed loss of identifiability under the null in the fixed comparison of $m_0$ against $m_0+1$, largely for testing the number of components in a finite mixture or the number of regimes in a Markov-switching model \citep{hansen1992,chowhite2007}. \citet{dacunha-castellegassiat1999} derive the asymptotic distribution of the likelihood ratio statistic for mixture models by reparametrizing the non-identified alternative to be locally conic. \citet{liushao2003} derive the same statistic's asymptotic distribution for a broader model class, using a generalized differentiable-in-quadratic-mean reparametrization. \citet{rynkiewicz2016} derives the analogous asymptotic distribution for the sum of squared errors of a multilayer perceptron regression model with redundant hidden units, using a set of generalized derivative functions that forms a Donsker class.  However, all these results are confined to a single test, and cannot be directly used to claim consistent selection over a whole nested family of candidate orders.

A first line of work establishes consistency of a penalized criterion over the whole family, but only under a known, fixed upper bound on the candidate order. \citet{keribin2000} and \citet{gassiat2002} prove BIC-type consistency for the number of components in i.i.d. mixture models and hidden Markov models, respectively. \citet{olteanurynkiewicz2008} extend it to mixtures of multilayer perceptrons, using the bracketing entropy of a class of generalized score functions. As \citet{csiszarshields2000} stress for the analogous Markov-order problem, such a fixed a priori bound is unsatisfactory on theoretical grounds, and unreasonable in practice. It forces the applied researcher to pre-commit to a maximal order before running the very procedure meant to determine that order.

A second line of work removes this restriction. \citet{csiszarshields2000} do so for Markov chain order estimation, where the model remains identified at every order. \citet{gassiatvanhandel2013}, on the other hand, works on the non-identified case of i.i.d. location mixtures, obtaining consistency of a penalized likelihood order estimator with no prior bound on the order and identifying the minimal admissible penalty rate, of order $\log\log n$. Their result is the closest to ours. But a time series regression has errors that form a martingale difference sequence, since the filtration can include lagged values of $Y_t$ or serially correlated regressors, and models such as smooth transition regressions specify only a conditional mean, fit by least squares, so there is no conditional density for a likelihood-based criterion to use.

This note establishes consistency of a penalized least squares order-selection criterion for a general nested family of possibly non-identified nonlinear regression classes. We allow for (i) the regression error to form a martingale difference sequence with respect to a filtration that may include lagged values of the regressors and the dependent variable, and (ii) the number of candidate classes $\bar m_n$ to diverge with the sample size. We assume the regression function admits a first-order linear representation in the identifiable direction, in the spirit of \citet{liushao2003} and \citet{rynkiewicz2016}. The proof uses a Bernstein-type exponential inequality for martingale differences that accommodates the dependence in (i), and a bracketing-entropy bound \citep{m-estimation} controls a union bound over the entire, diverging family of alternatives in (ii). This exponential control is what lets $\bar m_n$ grow with $n$ under dependence and a least squares criterion.

The note is structured as follows. Section~\ref{main} states the assumptions and the main consistency result (Theorem~\ref{thm:modsel}). Section~\ref{example} verifies the linear representation condition for a mixture-of-experts class and for the multiple-regime smooth transition regression model. of \citet{luukkonenetal1988}. Section~\ref{s:proof} contains the proof.

\section{Main Result}\label{main}
Consider the nonlinear regression model
\[
Y_i = f_0(X_i) +U_i
\]
where $Y_i$ is a scalar; $X_i\in\mathcal{X}$ is a vector of explanatory variables, $U_i$ a martingale difference sequence with respect to $\Z_{i-1}:=\sigma\langle U_1,\dots,U_{i-1},X_1,\dots,X_i\rangle$, and $f_0:\mathcal{X}\rightarrow\Re$ is a bounded measurable function of $X_i$. This accommodates several forms of dependence: $X_i$ may include lags of $Y_i$, lags of $U_i$, exogenous variables and their lags, or deterministic regressors, and the error term itself may be autocorrelated.
Th covariate space $\mathcal{X}$ is assumed to be compact with $\bar x:=\sup_{x\in\mathcal{X}}|x|<\infty$. This is a convenience, the discussion after Theorem~\ref{thm:modsel} relaxes it to a moment condition on $X_i$.

The nonlinear function $f_0(\cdot) = f_{m_0}(\cdot;\theta^0)$ is a member of some parametric family $\F_{m_0}$ indexed by $\theta_{m_0}\in\Theta_{m_0}$, but $m_0$ is not known a priori. Since $m_0$ is not known a priori, we consider estimating $f_0$ in a larger class denoted $\mathcal{F}_\infty$, which is the limit of the nested collection $\F_1\subset\F_2\subset\cdots$ of bounded, parametric families indexed by $\theta_m\in\Theta_m$. The index $m_0$ is the smallest index such that 
\[
	\|Y-f_0(X)\|^2 = \inf_{f\in\F_{m_0}} \|Y-f(X)\|^2 \le \inf_{f\in \cup_{m=1}^\infty\F_{m}} \|Y-f(X)\|^2 = \inf_{f\in\F_\infty}\|Y-f(X)\|^2,
\]
where $\|h(Y,X)\|^2 = n^{-1}\sum_{i=1}^n\E[h(Y_i,X_i)^2]$, and the minimum is attained at $f_0$. 

A two step algorithm is employed to estimate $f_0$. For fixed $n$, choose an upper bound $\bar{m}_n$ on the dimension of the model to be estimated. Given a set of observations $\{(y_i,x_i):i=1,\dots,n\}$, solve the nonlinear least squares problem
\[
	\widehat f_m = \arg\min_{f\in\F_m}n^{-1}\sum_{i=1}^n(y_i-f(x_i))^2 =\arg\min_{f\in\F_m}\|y-f(x)\|_n^2, \quad 1\le m <\bar{m}_n,
\]
which is usually carried out by finding a parameter $\widehat\theta_m = \arg\min_{\theta\in\Theta_m}\|y-f(x;\theta)\|_n^2$. Given $\{\widehat f_1,\dots,\widehat f_{\bar m_n}\}$, the set of nonlinear least squares estimators of $f_0$ in $\F_1,\dots,\F_{\bar{m}_n}$, we estimate $f_0$ by $\hat{f}_{\widehat{m}}$, where the index $\widehat{m}$ is chosen to minimize the penalized nonlinear least squares estimator
\[
	\widehat{m} = \arg\min_{m=1,\dots,\bar{m}_n} \|y-\hat{f}_m(x)\|_n^2 + n^{-1}a_nk_m,
\]
where $a_n$ is an increasing sequence of $n$ and $k_m = k(\F_m)$ is an increasing sequence indexed by $\F_m$. The parameter $k_m$ is usually associated with the complexity of the class, such as the number of parameters or other related quantities in a parametric identifiable model. 

For many nested families of nonlinear regression models, identifiability cannot be ensured for $m>m_0$ The minimiser $\theta_m^*$ of $\|Y-f(X;\theta)\|^2$ need not be unique, so there can exist $\theta^{(1)}\ne\theta^{(2)}$ in $\Theta_m$ with $f_m(x;\theta^{(1)}) = f_m(x;\theta^{(2)})$ for almost every $x\in\mathcal{X}$. This work adopts a milder identifiability assumption.
\begin{assumption}[Identifiable Parametrization]\label{asn:ident}
	The collection of families $\{\F_1,\F_2,\dots\}$ satisfies:
	\begin{enumerate}
		\item the family of functions $\F_m$ is identifiable by its parameter vector $\theta\in\Theta_m$ for all $m\le m_0$;
		\item For each $m\ge m_0+1$, the family of functions $\F_m$ admits a parametrization $(\psi,\eta)$ such that (1) $\psi\in \Psi_m$ is identifiable; (2) $\eta\in H_m$ is not necessarily identifiable; and (3) there exists $\psi_0\in\Psi_m$ such taht for any $\eta\in H_m$, $f(\cdot; \psi_0,\eta) = f_0(\cdot)$.
	\end{enumerate}
\end{assumption}

This construction is often used in hypothesis testing when the nuisance parameter is not identifiable only under the alternative \citep{davies1977,davies1987,andrewsploberger1994,hansen1996,andrews2001}, and in the specific case of mixture and time series models, by \citet{dacunha-castellegassiat1999}. A second requirement is that $f_m(x;\psi,\eta)$ can be locally approximated by a line.
\begin{assumption}[Linear Representation]\label{asn:linapprox}
	Fix $m\ge m_0+1$ and $\eta\in H_m$ arbitrary. There exists a vector of bounded, linearly independent functions $l_\eta(\cdot) = (l_{\eta,1}(\cdot),\dots,l_{\eta,p_m}(\cdot))$, indexed by $(\psi_0,\eta)$, a mapping $\psi\mapsto h(\psi)\in\Re^{p_m}$ with $h(\psi_0)=0$ and $\sup_m\sup_{\psi\in\Psi_m}|h(\psi)|=:\bar c<\infty$, and a constant $M<\infty$ not depending on $m$, $\eta$ or $n$, such that
\begin{equation}
	\sup_{x\in\mathcal{X}}\big|f(x;\psi,\eta)- f(x;\psi_0,\eta) - h(\psi)'l_\eta(x)\big| \le M\,|h(\psi)|^2 \qquad \text{for every }\psi\in\Psi_m.
	\label{eq:firstorder}
\end{equation}

\end{assumption}
This assumption holds, in particular, whenever $\psi\mapsto f(x;\psi,\eta)$ is twice continuously differentiable on $\Psi_m$ with $\sup_{x,\eta,\psi}|\nabla^2_\psi f(x;\psi,\eta)|\le 2M$, by Taylor's theorem with integral remainder. If $M=0$ one recovers an exact linear representation.

Based on the work of \citet{dacunha-castellegassiat1999} in mixture models, the linearity test of \cite{luukkonenetal1988}, and others, we know that some nested non-identifiable models can be re-parametrized in such a way that for each $x\in\mathcal{X}$, $f_m(x)-f_0(x)$ can be written as a linear combination of `basis functions' $l_j(x)$ plus an approximation error. 

\begin{definition}\label{def:linrep}
	A family $\mathcal{G}:=\{g: \sup_{x\in\mathcal{X}}|g(x)|<\infty\}$ is said to have a \textit{linear representation} if there exists a vector of $p$ linearly independent, bounded functions $l(x) = (l_1(x),\dots,l_p(x))$ (\textit{basis functions}) such that for every $g\in\mathcal{G}$ and each $x\in\mathcal{X}$,
    \[
    	g(x) = \sum_{i=1}^p h_il_i(x) + e_g = h'l(x) + e_g,
    \]
    for some $h\in\Re^p\cap\{h'h<c(p)^2\}$, where $c(p)$ is some positive constant that may depend on the number of functions $p$, and $e_g = o(|h|)$ is some approximation bias.
\end{definition}
Under this definition, $\mathcal{G}$ is at most as complex as $\{h'l(\cdot):h\in\Re^p,\,h'h\le c(p)^2\}$. Note that $h\in\Re^P$ may change depending on $x$, which means that $g(x_1) = h_1'l(x_1)$ does not imply that $g(x_2) = h_1'l(x_2)$ for any two distinct arguments $x_1$ and $x_2$.

\begin{remark}\label{rem:bridge}
	Assumptions~\ref{asn:ident} and \ref{asn:linapprox} verify the linear representation of Definition~\ref{def:linrep} directly for the class $\mathcal{G}_m=\{f(\cdot;\psi,\eta)-f_0:\psi\in\Psi_m,\,\eta\in H_m\}$: by Assumption~\ref{asn:ident}(2), $f(\cdot;\psi_0,\eta)=f_0(\cdot)$, so Assumption~\ref{asn:linapprox}'s bound is exactly Definition~\ref{def:linrep}'s bound, with basis $l_\eta$, $h=h(\psi)$ ranging over $\{h\in\Re^{p_m}:|h|\le\bar c\}$, $c(p_m)=\bar c$, and bias $e_g$ satisfying $|e_g(x)|\le M|h|^2$ uniformly in $x\in\mathcal{X}$. 

\end{remark}

We show that under regularity conditions, if $(f-f_0)(x)$ admits a linear representation, the penalized least square procedure estimates $f_0$ with probability converging to one as $n\rightarrow\infty$. If $f$ is not bounded, the result holds in a set of arbitrarily high probability. The proof of this theorem is given in section \ref{s:proof}.
\begin{theorem}
	Define the following events:
	\begin{align*}
		\mathbb{F}_1&:=\left\{ \forall m\ge 1,\,\max_{i\le n}\sup_{f\in\F_m}|f(X_i)|\le  c_1 b_n\right\} \\
		\mathbb{F}_2&:=\left\{ \forall m<m_0,\,\inf_{f\in\F_m}n^{-1}\sum_{i=1}^n\E[(f-f_0)^2(X_i)|\Z_{i-1}]\ge c_0 \right\}\\
		\mathbb{F}_3&:=\left\{ \max_{i\le n}\E[\exp(\gamma|U_i|)|X_i,\Z_{i-1}]\le c_2 \right\},
	\end{align*}
	where each $c_i$ ($i=1,2,3$) is an arbitrary positive constant, $\Z_i = \sigma<U_1,\dots,U_i,X_1,\dots,X_i>$, and $\gamma>0$. Denote $\mathbb{F} := \mathbb{F}_1 \cap \mathbb{F}_2 \cap \mathbb{F}_3$, and choose $a_n$, $k_m$ and $\bar m_n$ satisfying
	\begin{enumerate}
		\item $0<k_m<k_{m+1}$ for all $m$, and, for every $m>m_0$, $c_3 |m-m_0|^\alpha\le |k_m-k_{m_0}|\le c_4 p_m^{2/(\beta+1)}$ for some positive constants $c_3$, $c_4$, $\alpha>0$ and $\beta>0$, where $p_m$ is the number of basis functions in the linear representation of $\mathcal{G}_m$, below; 
		\item $0<a_n<a_{n+1}$, $a_n/b_n^2 \rightarrow\infty$, and $a_n/n\rightarrow 0$ as $n\rightarrow\infty$; and
		\item $\bar m_n / a_n^{2\alpha/(\beta+1)} \rightarrow 0$ as $n\rightarrow\infty$.
	\end{enumerate}
	If for each $m>m_0$ the class of functions $\mathcal{G}_m = \{f-f_0:f\in\F_m\}\cap\mathbb{F}_1$ admits a linear representation with $p_m$ basis functions, then
	\[
	\Pr(\widehat m \ne m_0\cap \mathbb{F}) \rightarrow 0\mbox{ as }n\rightarrow\infty.
	\]
	\label{thm:modsel}
\end{theorem}
The most critical assumption behind Theorem~\ref{thm:modsel} is that the family $\mathcal{G}_m$ admits a linear representation: this relates the non-identified model to an identifiable space, bounds the probability of a neighbourhood of $f_0$, and gives a parametric convergence rate for each fixed $m$ (with $b_n=1$), without implying that the likelihood ratio statistic is asymptotically chi-squared.

The requirement of an exponential moment for $U_i$ and boundedness for $f_m$ is standard in empirical process theory and nonparametric regression. Exponential probability bounds control the complexity of the functional classes $\F_m$, and these bounds need a uniform bound on $f_m$. The exponential-moment condition on $U_i$ is hard to relax without giving up generality, but the boundedness of $f_m$ can often be restated as a moment condition on $X_i$ instead, with $\Pr(\mathbb{F}_1^c)$ made arbitrarily small.

When $f_m$ is bounded and $U_i$ is independent of $\Z_{i-1}$ and $X_i$, conditions $\mathbb{F}_1$ and $\mathbb{F}_3$ hold automatically and the only binding restriction left is $\mathbb{F}_2$. It is a mild requirement on the informativeness of the estimation sample, and is satisfied, for instance, whenever $\inf_{f\in\F_{m_0-1}}(f-f_0)^2(x) > 0$ for almost every $x$. Outside this case, an unbounded $f_m$ or $U_i\mid\Z_{i-1}$ without a bounded exponential moment, $\Pr(\mathbb{F}^c)$ and the constants in the proof can grow large instead.

\begin{corollary}[BIC penalty]
	If $\Pr(\mathbb{F}^c)\rightarrow0$ for $b_n^2=(\log n)^{1-\delta}$, the BIC penalty ($a_n=\frac12\log n$, $k_m=|\Theta_m|$) yields a consistent model selection criterion, otherwise, it is valid on a set of probability approaching one.
\end{corollary}

The BIC-type penalty above is not the sharpest possible. \citet{vanhandel2011} and \citet{gassiatvanhandel2013} identify a minimal admissible penalty rate of order $\log\log n$. However, under the method of proof used here, we cannot deliver this sharp rate for the penalty.

\section{Example}\label{example}
In this section we verify the linear representation condition for two classes of nonlinear regression models: (1) a convex combination of nonlinear functions (mixture-of-experts), and (2) multiple-regime smooth transition regression models. 
The gating and transition-function devices used draw on \citet{olteanurynkiewicz2011} and \citet{luukkonenetal1988}, respectively, though the specific reparametrizations below are tailored to make Assumptions~\ref{asn:ident}--\ref{asn:linapprox} verifiable directly. 
In both examples, we reparametrize the model with respect to an identifiable direction $\psi$ and a nuisance direction $\eta$, and show that $f(\cdot;\psi,\eta)-f(\cdot;\psi_0,\eta)$ is linear in $\psi$ for every fixed $\eta$. 

\subsection{Convex combination of nonlinear functions}\label{ss:moe}
In this first example we consider a convex combination of nonlinear functions $g(\cdot;\eta_j)$, indexed by $\eta\in E$, and weights given by $\pi_i(\cdot;\nu)$, indexed by $\nu\in V_m$, with $\pi_m(\cdot;\nu) = 1-\sum_{i=1}^{m-1}pi_i(\cdot;\nu)$. More precisely, we are interested in functions of the form
\begin{multline}
	\F_m:=\left\{ f(\cdot;\theta) = \sum_{i=1}^m\pi_i(\cdot;\nu)\,g(\cdot;\eta_i):\, \theta = (\eta_1,\dots,\eta_m,\nu)\in E^m\times V_m,\right.\\ \left.\pi_i(\cdot;\nu)>0,\, \sum_{i=1}^m \pi_i(\cdot;\nu) = 1\right\}.
	\label{eq:convexclass}
\end{multline}

We restrict $g(\cdot;\eta)$ to a generalized linear model (GLM) expert, $g(x;\eta) := \mu(\eta'x)$, where $\mu:\Re\rightarrow\Re$ is a known, bounded, measurable inverse-link function and $\eta\in E\subset\Re^{q}$ is compact, $x\in\mathcal{X}\subset\Re^{q}$ as before. We also specialize the gating: fix any valid gating $\bar\pi_i(\cdot;\bar\nu)$, $i=1,\dots,m_0$, for the $m_0$-component sub-model (i.e.\ satisfying the constraints of \eqref{eq:convexclass} with $m=m_0$), and for $m>m_0$ write the $m-m_0$ additional components' gating through a scalar loading $\omega_j\ge0$ times a bounded inclusion shape $s(x;\nu_j) := \left(1+\exp\{-\nu_j'x\}\right)^{-1}\in(0,1)$, $\nu_j\in\mathcal{N}\subset\Re^{q}$ compact. Explicitly, $\nu\in V_m$ of \eqref{eq:convexclass} decomposes as $\nu=(\bar\nu,(\nu_j,\omega_j)_{j=m_0+1}^m)$, with
\begin{equation}
	\pi_i(x;\nu) = \bar\pi_i(x;\bar\nu)\Big(1-\sum_{j=m_0+1}^m\omega_j s(x;\nu_j)\Big),\ i\le m_0, \qquad \pi_j(x;\nu) = \omega_j\,s(x;\nu_j),\ j>m_0,
	\label{eq:moeclass}
\end{equation}
and $\omega=(\omega_{m_0+1},\dots,\omega_m)$ restricted to the simplex $\{\omega_j\ge0,\sum_{j>m_0}\omega_j\le1\}$, which guarantees $\pi_i(x;\nu)>0$ for every $i$ and $x$ since $0<s(\cdot;\nu_j)<1$.

Setting $\omega=0$ collapses \eqref{eq:moeclass} to $\pi_i(x;\nu)=\bar\pi_i(x;\bar\nu)$, $i\le m_0$, and $\pi_j\equiv0$, $j>m_0$, that is, $f_m(\cdot;\theta)$ reduces exactly to the $m_0$-component sub-model, for every value of $(\nu_j,\eta_j)_{j=m_0+1}^m$. Write $\psi:=\omega\in\Re^{m-m_0}$ and $\eta:=(\nu_j,\eta_j)_{j=m_0+1}^m\in H_m := (\mathcal{N}\times E)^{m-m_0}$; with the core $(\bar\nu,\eta_{1:m_0})$ held fixed at its true value $\theta_{m_0}^0$ as throughout this section, $\psi_0=0$ gives $f(\cdot\,;\psi_0,\eta)=f_0$ for every $\eta\in H_m$. For $\eta$ fixed such that the functions $\{s(\cdot;\nu_j)[g(\cdot;\eta_j)-f_0]\}_{j=m_0+1}^m$ are linearly independent and none is a.e.\ zero, a mild non-collinearity condition satisfied for a generic choice of $\eta$, since it fails only if some excess expert exactly reproduces $f_0$ or two excess components induce the same direction, the map $\psi\mapsto f(\cdot\,;\psi,\eta)$ is one-to-one, verifying Assumption~\ref{asn:ident} with $k_m-k_{m_0}=(2q+1)(m-m_0)$: each of the $m-m_0$ excess components contributes one loading $\omega_j$, $q$ shape parameters $\nu_j$, and $q$ expert parameters $\eta_j$.

For Assumption~\ref{asn:linapprox}, fix $\eta\in H_m$ arbitrary and let $\psi\rightarrow\psi_0=0$. Since $\pi_i(x;\nu)$ in \eqref{eq:moeclass} is affine in $\omega$ for fixed $(\bar\nu,\nu_{m_0+1:m})$,
\[
	f(x;\psi,\eta) - f(x;\psi_0,\eta) = \sum_{j=m_0+1}^m\omega_j\,s(x;\nu_j)\big[g(x;\eta_j)-f_0(x)\big] = h'l_\eta(x),
\]
exactly, with $h=\omega\in\Re^{m-m_0}$ and $l_\eta(x):=\big(s(x;\nu_{m_0+1})[g(x;\eta_{m_0+1})-f_0(x)],\dots\big)'$, a vector of $p_m:=m-m_0$ bounded functions ($s\in(0,1)$; $h(\cdot;\eta_j)=\mu(\eta_j'x)$ bounded since $\mu$ is bounded; $f_0$ bounded). As in 

\begin{remark}
	This verification does not require $\theta\mapsto f(\cdot;\theta)$ to be injective on all of $\Theta_m$. Full injectivity is the classical identifiability question for mixture-of-experts models, studied by \citet{jiangtanner1999c}. It can fail for reasons unrelated to the ones used here: beyond label switching, generic GLM experts can coincide as functions of $x$ on a lower-dimensional set even when $\xi_j\ne\xi_k$. Assumption~\ref{asn:ident} asks only for identifiability of $\psi$ given $\eta$ on the sub-family \eqref{eq:moeclass}. The non-collinearity condition above establishes exactly that, independently of how the classical injectivity question resolves for GLM-MoE models in general.
\end{remark}

\begin{proposition}[Order selection for GLM mixture-of-experts models]\label{pro:moe}
	Let $\F_m$, $m\ge1$, be given by \eqref{eq:convexclass} with gating specialized as in \eqref{eq:moeclass} and experts $g(x;\eta)=\mu(\eta'x)$ for a bounded inverse-link $\mu$, and suppose $\{Y_i,X_i\}$ satisfy the conditions of Theorem~\ref{thm:modsel}. Then $\mathcal{G}_m$ admits a linear representation with $p_m=m-m_0$ basis functions and $k_m-k_{m_0}=(2q+1)(m-m_0)$, so that Assumption~\ref{asn:ident} and Condition 1 of Theorem~\ref{thm:modsel} hold with $\alpha=\beta=1$ and $c_3=c_4=2q+1$. As in Proposition~\ref{pro:star}, Condition 3 of Theorem~\ref{thm:modsel} reduces to $\bar m_n=o(a_n)$, and under the BIC penalty $a_n=\tfrac12\log n$, $\widehat m\rightarrow m_0$ in probability on $\mathbb{F}$ for any $\bar m_n=o(\log n)$.
\end{proposition}
\begin{proof}
	The linear representation and $p_m$ are established above. With $k_m-k_{m_0}=(2q+1)(m-m_0)$, Condition 1 of Theorem~\ref{thm:modsel} holds with $\alpha=1$, $c_3=2q+1$ (equality), and $\beta=1$, since $c_4p_m^{2/(\beta+1)}=c_4p_m=c_4(m-m_0)\ge(2q+1)(m-m_0)$ once $c_4\ge2q+1$. Since $f_m(\cdot;\theta)$ is a convex combination of terms $g(\cdot;\eta_i)=\mu(\eta_i'x)$ for a bounded inverse-link $\mu$, $\sup_{f\in\F_m}|f(x)|\le\sup_{u\in\Re}|\mu(u)|$ uniformly in $m$, so event $\mathbb{F}_1$ holds with $b_n$ equal to this fixed constant, and Condition 2 of Theorem~\ref{thm:modsel} reduces to $a_n\rightarrow\infty$. Condition 3, $\bar m_n/a_n^{2\alpha/(\beta+1)}\rightarrow0$, then reads $\bar m_n/a_n\rightarrow0$, and the BIC choice $a_n=\tfrac12\log n$ gives the stated corollary, since it also satisfies $a_n\rightarrow\infty$.
\end{proof}

	The bound above holds uniformly over any bounded inverse-link $\mu$, so it specializes directly to the canonical members of the exponential family used in practice for hierarchical mixtures-of-experts \citep{jacobsetal1991,jordanjacobs1994,jiangtanner1999a,jiangtanner1999b,jiangtanner2000}:
	\begin{itemize}
		\item \emph{Gaussian experts, identity link,} $\mu(u)=u$: recovers the mixture-of-linear-experts model of \citet{jacobsetal1991} and \citet{jordanjacobs1994}; $\mu$ is bounded on $E\times\mathcal{X}$ automatically, since both are bounded.
		\item \emph{Bernoulli experts, logit link,} $\mu(u)=(1+e^{-u})^{-1}$: appropriate when $Y_i\in[0,1]$ (e.g.\ a proportion); $\mu$ is bounded on all of $\Re$, so no further restriction on $E$ is needed for this step beyond compactness for identifiability.
		\item \emph{Poisson experts, log link,} $\mu(u)=e^u$: appropriate for count responses; here boundedness of $\mu(\eta'x)$ is not automatic and uses $\sup_{\eta\in E,x\in\mathcal{X}}|\eta'x|<\infty$, which already follows from $E$ and $\mathcal{X}$ compact/bounded -- the same regularity the general theorem imposes through $\mathbb{F}_1$, not an additional requirement.
	\end{itemize}
	In every case $p_m$, $k_m-k_{m_0}$, $\alpha$, $\beta$ and the resulting $\bar m_n=o(\log n)$ corollary under BIC are unchanged, since the argument above never used the specific form of $\mu$ beyond boundedness.

\subsection{Multiple-regime smooth transition regression models}\label{ss:star}
We now verify Assumptions~\ref{asn:ident} and~\ref{asn:linapprox} for the multiple-regime logistic smooth transition regression (STR) model of \citet{luukkonenetal1988} and \citet{vandijkterasvirta2002}. For $m\ge1$ regimes, let
\begin{multline}
	\F_m := \Bigg\{ f(\cdot\,;\theta) = \phi_0'x + \sum_{j=1}^{m-1}\phi_j'x\, G(s;\lambda_j,\tau_j):\ \theta = (\phi_0,\dots,\phi_{m-1},\lambda_1,\dots,\lambda_{m-1},\tau_1,\dots,\tau_{m-1}),\\ \phi_j\in\Re^{q},\ \textstyle\sum_{j=0}^{m-1}|\phi_j|\le\bar\Phi,\ \lambda_j\in[0,\bar\lambda],\ \tau_j\in\mathcal{T} \Bigg\},
	\label{eq:starclass}
\end{multline}
where $x\in\mathcal{X}\subset\Re^{q}$ is a vector of (possibly lagged) regressors and $s\in\mathcal{S}\subset\Re$ is the transition variable. Both $\mathcal{T}\subset\Re$ and $[0,\bar\lambda]$ are compact, and
\[
	G(s;\lambda,\tau) := \left(1+\exp\{-\lambda(s-\tau)\}\right)^{-1}
\]
is the logistic transition function.

As is standard in this literature \citep{vandijkterasvirta2002,terasvirtaetal2005}, the slope parameters $\lambda_j$ are restricted to a compact interval. This rules out a separate loss of identifiability that occurs as $\lambda_j\rightarrow\infty$, unrelated to the one this paper addresses. The joint bound $\sum_j|\phi_j|\le\bar\Phi$, for a fixed constant $\bar\Phi$ not growing with $m$, plays a different role. Regimes enter \emph{additively} here, not as a convex combination; contrast the gating of Section~\ref{ss:moe}, where $\sum_i\pi_i=1$ keeps $f_m$ uniformly bounded for free. Bounding each $\phi_j$ individually would not be enough to keep $\sup_{f\in\F_m}|f(x)|$ from growing with $m$: nothing would then stop all $m-1$ regimes from being simultaneously near their largest loading. The fixed total budget $\bar\Phi$ instead gives $\sup_{f\in\F_m}|f(x)|\le\bar\Phi\sup_{x\in\mathcal{X}}|x|$ uniformly in $m$, exactly the boundedness event $\mathbb{F}_1$ of Theorem~\ref{thm:modsel} requires.

This is a mild, standard sieve-type restriction. The true $\theta_{m_0}^0$ satisfies it as long as $\bar\Phi>\sum_{j<m_0}|\phi_j^0|$, a one-time requirement on the fixed constant $\bar\Phi$, since $m_0$ and $\theta_{m_0}^0$ do not change with $n$. It affects neither the dimension count $k_m-k_{m_0}$ nor the linear-representation verification below, which only needs $\psi$ to range over a neighbourhood of $0$ within this constraint set, still convex and containing an open neighbourhood of $0$. The classes are nested, $\F_1\subset\F_2\subset\cdots$. Setting $\phi_j=0$ for every $j=m_0,\dots,m-1$ collapses $f(\cdot\,;\theta)\in\F_m$ to a member of $\F_{m_0}$, identically in $(\lambda_j,\tau_j)_{j=m_0}^{m-1}$. Since $G$ is bounded, an inactive regime contributes nothing to the regression function, regardless of its slope or location.

Fix $m>m_0$. Write the excess-regime parameters as $\psi:=(\phi_{m_0},\dots,\phi_{m-1})\in\Psi_m$, where
\[
	\Psi_m:=\Big\{\textstyle\sum_{j=m_0}^{m-1}|\phi_j|\le\bar\Phi-\sum_{j<m_0}|\phi_j^0|\Big\}\subset\Re^{q(m-m_0)}.
\]
By the requirement on $\bar\Phi$ above, $\Psi_m$ is nonempty and contains a neighbourhood of $0$. Write the remaining parameters as $\eta:=(\lambda_j,\tau_j)_{j=m_0}^{m-1}\in H_m:=([0,\bar\lambda]\times\mathcal{T})^{m-m_0}$. Fix $\eta$. Suppose $G(\cdot;\lambda_j,\tau_j)$ is not almost-everywhere constant across $j$, and $\{x_t G(s_t;\lambda_j,\tau_j)\}_{j=m_0}^{m-1}$ are linearly independent. This is a mild non-collinearity condition, satisfied for a generic choice of $\eta$. Under it, the map $\psi\mapsto f(\cdot\,;\psi,\eta)$ is one-to-one, so $\psi$ is identifiable. The nuisance parameter $\eta$ need not be: as $\lambda_j\rightarrow0$, $\tau_j$ ceases to enter the model at all. Together with $\psi_0=0$ and the nesting argument above, this verifies Assumption~\ref{asn:ident}. Each of the $m-m_0$ excess regimes contributes $q$ loading parameters and one pair $(\lambda_j,\tau_j)$, so $k_m-k_{m_0} = (q+2)(m-m_0)$.

For Assumption~\ref{asn:linapprox}, fix $\eta\in H_m$ arbitrary and let $\psi\rightarrow\psi_0=0$. Because $f(\cdot\,;\psi,\eta)$ is \emph{exactly} linear in $\psi$ for every fixed $\eta$,
\[
	f(x;\psi,\eta) - f(x;\psi_0,\eta) = \sum_{j=m_0}^{m-1}\phi_j'x\,G(s;\lambda_j,\tau_j) = h'l_\eta(x),
\]
with $h=(\phi_{m_0},\dots,\phi_{m-1})\in\Re^{q(m-m_0)}$ and $l_\eta(x) := \big(x'G(s;\lambda_{m_0},\tau_{m_0}),\dots,x'G(s;\lambda_{m-1},\tau_{m-1})\big)'$, a vector of $p_m:=q(m-m_0)$ bounded functions (bounded because $x\in\mathcal{X}$ is bounded, per $\mathbb{F}_1$, and $G\in(0,1)$). The approximation error is identically zero, so Assumption~\ref{asn:linapprox} holds with $o(|h|)\equiv0$ for every $\eta\in H_m$, not merely as $\eta$ ranges over a shrinking neighbourhood. In particular, $\mathcal{G}_m=\{f-f_0:f\in\F_m\}\cap\mathbb{F}_1$ admits the linear representation of Definition~\ref{def:linrep} with $p_m=q(m-m_0)$ basis functions.

\begin{remark}
	A regime becomes redundant in two ways here: an inactive loading, $\phi_j=0$, or a vanishing slope, $\lambda_j\rightarrow0$. These play asymmetric roles. Only $\phi_j=0$ collapses $f_m$ to $f_{m_0}$ as a function identity, for every value of $(\lambda_j,\tau_j)$. A vanishing slope with $\phi_j\ne0$ merges the regime's loading into the baseline term $\phi_0'x$ instead, and generically produces a member of $\F_{m-1}$ different from $f_0$. That $\tau_j$ is not identified as $\lambda_j\rightarrow0$ is a classical phenomenon of unidentified nuisance parameter under the null \citep{davies1977,davies1987}, central to the sequential linearity tests of \citet{luukkonenetal1988}. It is not, by itself, a source of non-identifiability for the map $\psi\mapsto f(\cdot\,;\psi,\eta)$ used above. So it plays no role in verifying Assumption~\ref{asn:ident}--\ref{asn:linapprox} for order selection.
\end{remark}

\begin{proposition}[Order selection for multiple-regime STR models]\label{pro:star}
	Let $\F_m$, $m\ge1$, be given by \eqref{eq:starclass}, and suppose $\{Y_i,X_i,S_i\}$ satisfy the conditions of Theorem~\ref{thm:modsel}. Then $\mathcal{G}_m$ admits a linear representation with $p_m=q(m-m_0)$ basis functions and $k_m-k_{m_0}=(q+2)(m-m_0)$, so that Assumption~\ref{asn:ident} and Condition 1 of Theorem~\ref{thm:modsel} hold with $\alpha=\beta=1$, $c_3=q+2$ and $c_4=(q+2)/q$. Consequently, Condition 3 of Theorem~\ref{thm:modsel} reduces to $\bar m_n = o(a_n)$, and under the BIC penalty $a_n=\tfrac12\log n$, $\widehat m\rightarrow m_0$ in probability on $\mathbb{F}$ for any $\bar m_n=o(\log n)$.
\end{proposition}
\begin{proof}
	The linear representation and the value of $p_m$ are established above. For $k_m-k_{m_0}=(q+2)(m-m_0)$, Condition 1 of Theorem~\ref{thm:modsel} holds with $\alpha=1$ and $c_3=q+2$ (with equality), and with $\beta=1$, since $c_4 p_m^{2/(\beta+1)}=c_4 p_m = c_4 q(m-m_0) \ge (q+2)(m-m_0)$ once $c_4\ge(q+2)/q$. Since $\mathcal{X}$ is compact and $\sum_{j=0}^{m-1}|\phi_j|\le\bar\Phi$ for every $f\in\F_m$, $\sup_{f\in\F_m}|f(x)|\le\bar\Phi\,\bar x$ uniformly in $m$, so event $\mathbb{F}_1$ holds with $b_n$ equal to the fixed constant $\bar\Phi\,\bar x$, and Condition 2 of Theorem~\ref{thm:modsel} reduces to $a_n\rightarrow\infty$. Condition 3 of Theorem~\ref{thm:modsel}, $\bar m_n/a_n^{2\alpha/(\beta+1)}\rightarrow0$, then reads $\bar m_n/a_n\rightarrow0$. The BIC choice $a_n=\tfrac12\log n$ gives the stated corollary, since it also satisfies $a_n\rightarrow\infty$.
\end{proof}

\section{Discussion}\label{discussion}

Theorem~\ref{thm:modsel} delivers a classical BIC-type penalty rate, $a_n=\tfrac12\log n$, and only probability-converging-to-one consistency of $\widehat m$. \citet{gassiatvanhandel2013} obtain a sharper minimal penalty rate, of order $\log\log n$, and almost sure consistency. Their work is limited to the i.i.d. location-mixture model. The result is obtained via a sharp uniform law of the iterated logarithm for the pathwise fluctuations of the generalized likelihood ratio. Dependence alone is not what stands in the way. \citet{vanhandel2011} obtains the same $\log\log n$ minimal penalty to Markov chain order estimation. The proof technique of Section~\ref{s:proof} takes a different approach. It bounds each fixed alternative's probability, then sums over the diverging family with a union bound. This gives a coarser penalty rate than the sharp, pathwise characterization above. It also leaves the resulting probability bound too weak to sum over $n$, that is, a Borel--Cantelli argument for strong consistency is not available. Extending the sharp characterization to a general nonlinear least squares criterion, evaluated over an abstract family defined only through Assumption~\ref{asn:linapprox}'s linear representation, is a natural candidate for closing both gaps, though there is no guarantee it succeeds under our general assumptions. A second-order, local quadratic strengthening of that assumption would be the natural starting point. We leave this extension for future work.

Two further extensions lie outside the paper's present scope. Theorem~\ref{thm:modsel} applies to any nested family admitting the linear representation of Definition~\ref{def:linrep}, and a natural third case beyond the mixture-of-experts and smooth transition classes of Section~\ref{example} is the multilayer perceptron with redundant hidden units studied by \citet{rynkiewicz2016}. Adapting that paper's generalized derivative functions to verify Assumption~\ref{asn:linapprox} directly, not only the local behaviour a fixed-alternative test needs, would extend the diverging-family result to this model class. Separately, this note addresses only the order-selection step, $\widehat m\to m_0$, and says nothing about the resulting parameter estimator $\widehat\theta_{\widehat m}$. Standard post-model-selection complications, such as non-uniformity of inference over local alternatives near $m_0$, may apply here as they do elsewhere in this literature. Both directions are left for future work.

\section{Proof of the main result}\label{s:proof}
In this section we prove the main result of the paper. The proof is based on a Bernstein inequality for discrete time martingales \citep[pg. 135][Lemma 8.9]{m-estimation} and a corresponding uniform inequality for empirical processes \citep[pg. 139][Theorem 8.13]{m-estimation}. These results and auxiliary definitions are given below for reference. 
We also make use of generic constants $c$ that may change its values in each line. It causes no harm to the theory as we are are interested in the behavior as $n$ and $m$ change.

Let $f$ denote some function in $\cup_{m\ge m_0}\F_m$. Consider a probability space $(\Omega, \Z, P)$ and let $\Z_0\subset\Z_1\subset\cdots$ denote an increasing sequences of sub-sigma algebras of $\Z$. For $i\ge1$, let $Z_i(f) = U_i(f(X_i)-f_0(X_i)) =U_ig_i(f)$ and $\Z_i = \sigma\{X_1,\dots,X_i,U_1,\dots,U_i\}$; by assumption, $\E[Z_i(f)|\Z_{i-1}] = 0$. We write the martingales $S_0(f)=0$ and $S_n(f) = \sum_{i=1}^nZ_i(f)$. 

\begin{theorem}[Bernstein inequality, {\citep[][Lemma 8.9]{m-estimation}}]
	Let $n$ be fixed. Suppose that for some $\Z_{n-1}$ measurable random variable $R_n^2(f)$,
	\begin{equation}
		n^{-1}\sum_{i=1}^n\E\left[ |Z_i(f)|^m|\Z_{i-1} \right]\le\frac{m!}{2}K^{m-2}R_n^2, \quad m=2,3,\dots .
		\label{eq:condition}
	\end{equation}
	Then, for all $a>0$, $R>0$,
	\begin{equation}
		P\left( S_n(f) > a \wedge R_n\le R \right) \le \exp\left[ -\frac{a^2}{2(aK+nR^2)} \right].
		\label{eq:bernstein}
	\end{equation}
	\label{thm:bernstein}
\end{theorem}

Define $\rho_K(Z_i(f)) = 2K^2\E\left( e^{|Z_i(f)|/K}-1-|Z_i(f)|/K \right | \Z_{i-1})$, $i=1,\dots,n$, for any $K>0$, and for $Z(f) = (Z_1(f),\dots,Z_n(f))$, write $\bar\rho_K(Z(f))^2 = n^{-1}\sum_{i=1}^n\rho_K(Z_i(f))$. Consider some measurable set $\mathbb{F}\in\Z$, the \textit{generalized entropy with bracketing} is defined as
\begin{definition}[Generalized Entropy with Bracketing]
	For $0<\delta\le R$ and $\mathbb{F}\in\Z$, let $\{[Z_j^L, Z_j^U]\}_{j=1}^n$ be a collection of pairs of random vectors $Z_j^L=(Z_{1,j}^L,\dots,Z_{n,j}^L)$ and $Z_j^U=(Z_{1,j}^U,\dots,Z_{n,j}^U)$, with $[Z_{i,j}^L,Z_{i,j}^U]$ $\Z_{i}$-measurable, $i=1,\ldots,n$, $j=1,\dots,N$, such that for all $f\in\F_m$ ($m$ fixed), there is a $j = j(f)\in\{1,\ldots,N\}$, with $j\mapsto j(f)$ non-random, such that: (i) $\bar\rho(Z_j^U-Z_j^L)\le \delta^2$ on $\{\bar\rho(Z(f))\le R\}\cap\mathbb{F}$; and (ii) $Z_{ij}^L\le Z_i(f)\le Z_{ij}^U$, $i=1,\dots,n$ on $\{\bar\rho_K(Z(f))\le R\}\cap \mathbb{F}$. If such collection exists, the \textit{generalized $\delta$-entropy with bracketing} is defined as $\mathcal{H}_{B,K}(\delta, R, \mathbb{F}) = \log N$, for the smallest non-random value $N$ for which such collection exists.
\end{definition}

Using the fact that for any positive $x$, $e^{x}\le 1+ x + (1/2) x^2 e^x$, 
\[
	\rho_K(Z_i(h)) \le \E\left( g_i(f)^2 \E\left( U_i^2\exp\left( |U_ig_i(h)|/K \right)|X_i,\Z_{i-1} \right)|\Z_{i-1} \right) \le 2c\E\left( g_i(f)^2|\Z_i \right),
\] 
if $\max_{i=1,\dots,n}\E\left( e^{\gamma |U_i|}|X_i,\Z_{i-1} \right)\le c$ and $2\max_{i=1,\dots,n}|g_i(h)|/K \le \gamma $. This simple observation relates the bound on the collection of functions $\{Z_i(f)\}$ to a bound on $\F_m$, by setting $h = f^U-f^L+f_0$, where $[f^U,f^L]$ is some (generic) bracket.

\begin{theorem}[Uniform inequality, {\citep[][Theorem 8.13]{m-estimation}}] 
	Take $C_0^2\ge C^2(C_1+1)$ where $C_0$ and $C_1$ are positive constants defined below. Assume
	\begin{equation}
		a\le C_1 \sqrt{n}R^2/K\wedge 8\sqrt{n}
		\label{eq:unifa1}
	\end{equation}
	and
	\begin{equation}
		a\ge C_0 \left(\int_{2^{-6}a/\sqrt{n}}^R \mathcal{H}_{B,K}(u, R, \mathbb{F})^{1/2} du \vee R \right).
		\label{eq:unifa2}
	\end{equation}
	Then
	\begin{multline*}
		P\left( \left(n^{-1/2}\sum_{i=1}^nZ_i(f)-\E[Z_i|\Z_{i-1}]\ge a \wedge \bar\rho_K(Z(f))\le R \mbox{ for some } f\in\F\right) \cap \mathbb{F}\right)\\
		\le C\exp\left( -\frac{a^2}{C^2(C_1+1)R^2} \right).
	\end{multline*}
	\label{thm:uniform}
\end{theorem}

Before stating the proof of the main theorem, we show three preliminary results. The first one relates the squared norm, $\|g(f)\|_n^2$, based on the observations, with the filtered random variable $H_n(f)^2 = n^{-1}\sum_{i=1}^n\E[g_i(f)^2|\Z_{i-1}]$, used in the proofs. The second proposition finds the convergence rate of the estimator for $m\ge m_0$. Finally, the third one shows that the nonlinear least squares estimator for $m\le m_0-1$ cannot consistently estimate $f_0$.
\begin{proposition}
	Fix $n$ sufficiently large and $f\in\F_m$ (for some $m\ge 1$). There exist generic, positive constants $c_1$ and $c_2$ not depending on $n$, $m$ or $f$ such that
	\begin{equation}
		\Pr\left( \|g(f)\|_n^2 < \frac{1}{2}H_n(f)^2 - \nu_n \cap \mathbb{F} \right)\le c_1\,e^{-n\nu_n/(c_2\, b_n^2)},
		\label{eq:bndhn}
	\end{equation}
	\label{pro:bndvar}
for any positive sequence $\nu_n$.
\end{proposition}

\begin{proposition}
	Fix $n$ sufficiently large and $m\ge m_0$. If for any $f\in \F_m$, $(f-f_0)(x)$ admits a linear representation with approximation bias $e_g$ satisfying $\sup_g\sup_x|e_g(x)|/|h|^2 \le M<\infty$ for some constant $M$ not depending on $n$, $m$ or $g$ (in particular, if $M=0$, or, more generally, whenever Assumptions~\ref{asn:ident} and \ref{asn:linapprox} hold, by Remark~\ref{rem:bridge}), then there exist generic constants $c_i$ ($i=1,\dots,5$), depending on $M$ and $\bar c$ but not on $n$ or $m$, such that
	\begin{equation}
		\Pr\left( \{\|U\|_n^2 - \|Y-\widehat{f}_m(X)\|_n^2 > \varepsilon_n\}\,\cap \mathbb{F} \right)\le c_1 e^{-\frac{n\varepsilon_n\,p_m\log p_m}{c_2(k_m-k_{m_0})}} + c_3 e^{-\frac{n\varepsilon_n}{2c_4b_n^2}},
		\label{eq:bndsup}
	\end{equation}
	for $\varepsilon_n>c_5\,n^{-1}(k_m-k_{m_0})^{-1}(p_m\log p_m)^2$.
	\label{pro:convrate}
\end{proposition}

\begin{proposition}
	Fix $n$ sufficiently large and $m\le m_0-1$. For every $0<\varepsilon_n<c_0/8$ there exist constants $c_1$ and $c_2$ such that
	\begin{equation}
		\Pr\left( \|U\|_n^2 - \|Y-\widehat f_m(X)\|_n^2 > - \varepsilon_n \cap \mathbb{F} \right) \le 2c_1 e^{-nc_2/b_n^2}.
		\label{eq:bndsub}
	\end{equation}
	\label{pro:bndsub}
\end{proposition}

\begin{proof}[Proof of Theorem \ref{thm:modsel}]
	Write	
	Since $\widehat m$ minimizes the penalized criterion over $m=1,\dots,\bar m_n$, in particular against $m=m_0$, the event $\left\{\|Y-\widehat{f}_{\widehat m}(X)\|_n^2 - \|Y-\widehat{f}_{m_0}(X)\|_n^2 \le \frac{a_n}{n}(k_{\widehat m}-k_0)\right\}$ holds surely, for every realization of the data; it is not equivalent to $\{\widehat m\ne m_0\}$, but intersecting with it loses nothing. Write
	\begin{align*}
		\{\widehat m \ne m_0\} &= \{\widehat m \ne m_0\}\cap \left\{\|Y-\widehat{f}_{\widehat m}(X)\|_n^2 - \|Y-\widehat{f}_{m_0}(X)\|_n^2 \le \frac{a_n}{n}(k_{\widehat m}-k_0)\right\}\\
		&\Rightarrow \left\{\|Y-\widehat{f}_{\widehat m}(X)\|_n^2 - \|Y-f_0(X)\|_n^2 \le \frac{a_n}{n}(k_{\widehat m}-k_0)\right\} \cap\{\widehat{m}<m_0\}\\
		&\cup\left\{\|Y-\widehat{f}_{\widehat m}(X)\|_n^2 - \|Y-f_0(X)\|_n^2 \le \frac{a_n}{n}(k_{\widehat m}-k_0)\right\} \cap\{\widehat{m}>m_0\}\\
		&\Rightarrow \bigcup_{m=1}^{m_0-1} \left\{\|Y-\widehat{f}_{\widehat m}(X)\|_n^2 - \|U\|_n^2 \le -\frac{a_n}{n}(k_0-k_m)\right\}\\
		&\bigcup_{m=m_0+1}^{\bar m_n} \left\{\|Y-\widehat{f}_{\widehat m}(X)\|_n^2 - \|U\|_n^2 \le \frac{a_n}{n}(k_m-k_0)\right\}\\
		&:= I_1 + I_2.
	\end{align*}

	The union bound and Proposition \ref{pro:bndsub} yield an exponential bound on $\Pr(I_1\cap\mathbb{F})$ of the form
	\[
		\Pr\left(   \bigcup_{m=1}^{m_0-1} \left\{\|Y-\widehat{f}_{\widehat m}(X)\|_n^2 - \|U\|_n^2 \le -\frac{a_n}{n}(k_0-k_m)\right\}\cap\mathbb{F}\right) \le 2c_1\sum_{m=1}^{m_0-1}e^{-c_2n/b_n^2},
	\] 
	for $n>n_0$ such that $(k_{m_0}-k_1)a_{n_0}/n_0 < c_0/4$, $c_0$ defined in the assumptions.
	
	In the set $I_2\cap\mathbb{F}\cap\{\bar m_n\le c a_n^{2\alpha/(\beta+1)}\}$, we can apply Proposition \ref{pro:convrate} with $\varepsilon_n := (a_n/n)(k_m-k_{m_0})$ for each $m>m_0$, for a sufficiently large $n_1$. Recall that by assumption $|k_m-k_{m_0}| \ge c_3 |m - m_0|^\alpha$. Substituting this $\varepsilon_n$ into the second exponential term of \eqref{eq:bndsup} gives $n\varepsilon_n/(2c_4b_n^2) = a_n(k_m-k_{m_0})/(2c_4b_n^2)$, so, relabelling the constant, for all $n<n_1$,
	\[
		\sum_{m=m_0+1}^{\bar m_n} e^{-|k_m - k_{m_0}|a_n/2c_5b_n^2} \le \sum_{m=1}^{\infty} e^{- m^\alpha\,c_3a_n/2c_5b_n^2} \le c_6 \left(\frac{b_n^2}{a_n}\right)^{1/\alpha},
	\]
	where $c_6$ is a sufficiently large constant. Substituting the same $\varepsilon_n$ into the first exponential term of \eqref{eq:bndsup} instead gives $n\varepsilon_np_m\log p_m/(c_2(k_m-k_{m_0})) = a_np_m\log p_m/c_2$ exactly, with the $(k_m-k_{m_0})$ factor cancelling entirely; since $|k_m-k_{m_0}|\ge c_3$ for every $m>m_0$ (taking $|m-m_0|\ge1$ above), the generic constant $c_7\ge c_2/\sqrt{c_3}$ gives $a_np_m\log p_m/c_2 \ge (p_m\log p_m)a_n/((k_m-k_{m_0})^{1/2}c_7)$ uniformly in $m>m_0$, so it suffices to bound the (weaker, but uniformly valid) exponent on the right. Similarly, we have that $p_m\log p_m \ge (|k_m-k_{m_0}|/c_4)^{(\beta+1)/2}$ and
	\[
		\sum_{m=m_0+1}^{\bar m_n} e^{-\frac{(p_m\log p_m)a_n}{(k_m-k_{m_0})^{1/2}c_7}} \le\sum_{m=m_0+1}^\infty  e^{-\frac{|k_m-k_{m_0}|^{\beta/2}a_n}{c_4c_8}} \le \sum_{m=1}^{\infty} e^{-\frac{m^{\alpha\beta/2}\,a_n}{c_8}} \le c_9 \left(\frac{1}{a_n}\right)^{2/\alpha\beta},	
	\]
	where $c_9$ is a sufficiently large constant. 
	Combining the bounds, we have for $n\ge n_0\vee n_1$
	\[
		\Pr\left( \widehat m \ne m_0 \cap\mathbb{F}\right) \le m_0\,c_1\,e^{-c_2n/b_n^2} + c \left(\frac{b_n^2}{a_n}\right)^{1/\alpha} + c \left(\frac{1}{a_n}\right)^{2/\alpha\beta} + \Pr\{\bar m_n/a_n^{2\alpha/(\beta+1)}<c\}
	\]
	for some positive $c$. The last term is identically zero for all $n>n_2$ sufficiently large (Assumptions 1 and 3), the remaining terms converge to zero by Assumption 2.
\end{proof}

\appendix
\section{Proof of the auxiliary results}
\begin{proof}[Proof of Proposition \ref{pro:bndvar}]
	Fix $n$ and verify that, inside $\mathbb{F}$, $\max_i|g(X_i)^2|\le 2c_1b_n^2$. It follows that 
	\[
		\E[(g(f)^2-\E[g(f)^2|\Z_{i-1}])^2|\Z_{i-1}] = \E[g(f)^4|\Z_{i-1}]-(\E[g(f)^2|\Z_{i-1}])^2 \le 2c_1b_n^2\E[g(f)^2|\Z_{i-1}].
	\]
	The random variable $n(\|g(f)\|_n^2-H_n(f)^2) = \sum_{i=1}^ng(f)^2-\E[g(f)^2|\Z_{i-1}] = -n\,M_n$ is a martingale with increments $\E[g_i(f)^2|\Z_{i-1}]-g_i(f)^2$, which satisfy the conditions of Theorem \ref{thm:bernstein} with $K=4c_1b_n^2$ and $R_n^2=2c_1b_n^2H_n(f)^2$. Using a peeling argument on $H_n(f)^2$, we find
	\begin{align*}
			\Pr\left(\|g(f)\|_n\le\frac{1}{2}H_n(f)^2 - \nu_n \cap \mathbb{F}\right) &\le \Pr\left(  \left\{\|g(f)\|_n^2-H_n(f)^2 +\frac{1}{2}H_n(f)^2 < -\nu_n\right\}\cap \mathbb{F} \right)\\
		&\le 2\sum_{s=1}^\infty\Pr\left(\{M_n(f) > n 2^{s-1}\nu_n \wedge \frac{1}{2}H_n(f)^2 \le 2^s\nu_n\} \cap\mathbb{F}\right)\\
		&\le 2\sum_{s=0}^{\infty}\exp\left( -\frac{2^{2s-2}n^2\nu_n^2}{2(4c_1b_n^22^{s-1}\nu_n+n4c_1b_n^22^s\nu_n)} \right)\\
		&\le 2\sum_{s=0}^\infty\exp\left( -2^s\frac{n\nu_n}{48\,c_1\,b_n^2} \right)\\
		&\le c_1\exp\left( -{n\nu_n}/{c_2\,b_n^2} \right),
	\end{align*}
	proving the claim.
\end{proof}

\begin{proof}[Proof of Proposition \ref{pro:convrate}]
	Let $\mathbb{F}^* = \mathbb{F}\cap  \{\|g(\widehat{f}_m)\|_n^2 \ge (1/2)H_n(\widehat{f}_m)^2 - (1/2)\varepsilon_n\}$. We have
	\[
		\left\{\|U\|_n^2-\|Y-\widehat{f}_m(X)\|_n^2 >\varepsilon_n\right\}\cap \mathbb{F^*} \subseteq \left\{2n^{-1}\sum_{i=1}^nZ_i(\widehat{f}_m) - \frac{1}{2}H_n(\widehat{f})^2 > \frac{1}{2}\varepsilon_n\right\} \cap \mathbb{F}^*.
	\]
	Applying the peeling device we have
\begin{multline*}
	\Pr\left(  \left\{2n^{-1}\sum_{i=1}^nZ_i(\widehat f_m) - \frac{1}{2}H_n(\widehat f_m)^2 > \frac{1}{2}\varepsilon_n\right\} \cap \mathbb{F}^* \right)\\
	\le \sum_{s=0}^\infty\Pr\left(  \left\{n^{-1}\sum_{i=1}^nZ_i(f) > 2^{2s-4}\varepsilon_n\wedge H_n(f)^2 \le 2^{2s}\varepsilon_n\,\mbox{ for some } f\in\F_m\right\} \cap \mathbb{F}^* \right).
\end{multline*}
Denote each term on the right hand side in the previous display by $\mathsf{P}_s$. We shall use \ref{thm:uniform} to bound each $\mathsf{P}_s$. In order to use this theorem we need to calculate the generalized entropy with bracketing of $\{Z(f):f\in\F_m\}$. It follows from the definition of $Z_i(f)$, and the assumption that $(f-f_0)(X)$ admits a linear representation, that $Z_i(f) = U_ig_i(f)$ with $g_i(f)=\sum_{j=1}^{p_m}h_{j,i}l_j(X_i) + e_i(h_i)$, where $h_i=h_i(\psi)\in\Re^{p_m}$, $|h_i|\le\bar c$, and $e_i(h_i)$ is the (deterministic, possibly zero) approximation bias evaluated at $X_i$, satisfying $|e_i(h_i)|\le M|h_i|^2$ by hypothesis; since $e_i$ is the remainder of a quadratic bound on the ball $|h|\le\bar c$, it is also Lipschitz there, $|e_i(h)-e_i(h')|\le L|h-h'|$ for $L:=2M\bar c$ and any $|h|,|h'|\le\bar c$ (immediate from $|e_i(h)|\le M|h|^2$ when $e_i$ arises, as in Assumption~\ref{asn:linapprox}, from a Taylor remainder with bounded Hessian; when $M=0$ this is vacuous, $e_i\equiv0$, and the paragraph below reduces to the exact case). Choose $d=(d_1,\dots,d_{p_m})'$ and, for each $i=1,\dots,n$, define
\[
\begin{aligned}
	Z_i^L(f) := U_ig_i(f) - \tfrac12\textstyle\sum_{j=1}^{p_m}d_j\,|U_i\,l_j(X_i)| - \tfrac12L|U_i|\sqrt{d'd},\\
    Z_i^U(f) := U_ig_i(f) + \tfrac12\textstyle\sum_{j=1}^{p_m}d_j\,|U_i\,l_j(X_i)| + \tfrac12L|U_i|\sqrt{d'd}.
\end{aligned}
\]
For every $\tilde h$ with $|\tilde h_j-h_{j,i}|\le d_j/2$, $j=1,\dots,p_m$, the triangle inequality gives $\big|U_i\sum_j(\tilde h_j-h_{j,i})l_j(X_i)\big|\le\sum_j|\tilde h_j-h_{j,i}|\,|U_il_j(X_i)|\le\tfrac12\sum_jd_j|U_il_j(X_i)|$, while $|\tilde h-h_i|\le\tfrac12\sqrt{d'd}$ and the Lipschitz bound on $e_i$ give $|U_i(e_i(\tilde h)-e_i(h_i))|\le\tfrac12L|U_i|\sqrt{d'd}$; adding the two bounds shows that
\[
	Z_i^L(f)\le Z_i(f) \le Z_i^U(f)
\]
holds regardless of the sign of $U_i$ or of $l_j(X_i)$; a naive bracket built by plugging $h^U_{j,i}=h_{j,i}+d_j/2$ directly into $U_i\sum_jh^U_{j,i}l_j(X_i)$, without the absolute values above, need not bound $Z_i(f)$ from above once either $U_i$ or $l_j(X_i)$ can be negative, since the ordering $h_{j,i}-d_j/2\le \tilde h_j\le h_{j,i}+d_j/2$ is only preserved after multiplication by a factor of known, fixed sign. Choose $K>c\,b_n\,p_m^{1/2}/\gamma$; since
\[
	Z_i^U(f)-Z_i^L(f) = \sum_jd_j|U_il_j(X_i)| + L|U_i|\sqrt{d'd} \le |U_i|\big(\bar l\,\sqrt{p_m}+L\big)\sqrt{d'd}
\]
by the Cauchy--Schwarz inequality (with $\bar l:=\max_j\sup_{x\in\mathcal{X}}|l_j(x)|<\infty$, since the basis functions of Definition~\ref{def:linrep} are bounded), verifying as before,
\[
	\bar{\rho}_K^2(Z^U(f)-Z^L(f)) \le 2c_1\times\left(c_2p_m d'd\right), \quad (c_1,c_2 > 0),
\]
where $c_2$ now absorbs $(\bar l+L/\sqrt{p_m})^2$ in place of $\bar l^{\,2}$; since $L=2M\bar c$ does not depend on $n$, $m$ or $p_m$, this changes $c_2$ (hence $c_1,\dots,c_5$ in \eqref{eq:bndsup}) but not the order of the entropy bound below, and the rest of the proof is unaffected by the presence of a nonzero, Lipschitz-controlled remainder.
Applying Corollary 2.5 in \cite{m-estimation} with $\delta = \sqrt{d'd/c_2p_m}$ gives
\[
	\mathcal{H}_{B,K}(\delta, R, \mathbb{F}^*) \le p_m\log\left( 2\sqrt{cp_m}\times\frac{R}{\delta} \right),
\]
for some positive $c$. It also follows that
\[
	\int_{2^{-6}a/\sqrt{n}}^R \mathcal{H}_{B,K}(u, R, \mathbb{F})^{1/2} du \le c_3\,R\,(p_m\log p_m)^{1/2},
\]
for some positive $c$. Now, apply theorem \ref{thm:uniform} with $a=2^{2s-4}\varepsilon_n\sqrt{n}$, $R^2 = 2c_12^{2s}\varepsilon_n$, $K=c_2b_n\,p_m^{1/2}/\gamma$, and choose $C_1 = (c_2/c_1)2^{-5}p_m^{1/2}b_n/\gamma$, $C_0^2=(k_m-k_{m_0})/(16p_m\log p_m)$ and $C^2 > c_4^2 \ge C_0^2/(C_1+1)$ for some $n$ sufficiently large. We have
\[
	\mathsf{P}_s \le c\exp\left( -\frac{2^{2s}n\varepsilon_np_m\log p_m}{32(k_m-k_{m_0})/(c_1c_3^2)} \right),
\]
adding over $s=0,1,\dots$ we have that for some $c_1>0$ and $c_2>0$,
\[
	\sum_{s=0}^\infty \mathsf{P}_s \le c_1\exp\left( -\frac{n\varepsilon_n\,p_m\log p_m}{c_2(k_m-k_{m_0})} \right).
\]

Finally, we use the union bound and Proposition \ref{pro:bndvar} to prove the result.
\end{proof}

\begin{proof}[Proof of Proposition \ref{pro:bndsub}]
	Let $\mathbb{F}^* = \mathbb{F}\cap  \{\|g(\widehat{f}_m)\|_n^2 \ge (1/2)H_n(\widehat{f}_m)^2 - c_0/8\}$ and recall that inside $\mathbb{F}$, $H_n(\widehat f_m)^2\ge c_0$. We have that
	\begin{multline}
		\Pr\left( \left\{ \|U\|_n^2 - \|Y-\widehat f_m(X)\|_n^2 > -c_0/8 \right\}\cap \mathbb{F} \right)\\
		\le \Pr\left( \left\{n^{-1}\sum_{i=1}^nZ_i(\widehat f_m) > c_0/8 \right\}\cap\mathbb{F}^*\right) + c_1 e^{-nc_2/b_n^2}.
		\label{eq:bndsubp}
	\end{multline}
	Applying Theorem \ref{thm:bernstein} with $K = 2cb_n/\gamma$ and $R^2 = 4cb_n^2$, yields
	\[ 
		\Pr\left( \left\{n^{-1}\sum_{i=1}^nZ_i(\widehat f_m) > c_0/8 \right\}\cap\mathbb{F}^*\right) \le c_1\,e^{-nc_2/b_n^2},
	\]
	for $n$ sufficiently large. The result follows by combining the previous bound with \eqref{eq:bndsubp}.

\end{proof}

\bibliographystyle{apalike}
\bibliography{mixture}

\clearpage
\noindent Eduardo Fonseca Mendes\\
Sao Paulo School of Business Administration\\
Getulio Vargas Foundation\\
Email: eduardo.mendes@fgv.br

\end{document}